\documentclass[12pt]{article}

\usepackage[backend=biber,style=numeric,]{biblatex}
\usepackage{amsfonts, amsmath}

\usepackage{amsthm}
\newtheorem{theorem}{Theorem}
\newtheorem{lemma}{Lemma}

\usepackage{booktabs}

\usepackage{etoolbox}
\preto\tabular{\setcounter{alglinecounter}{0}}
\newcounter{alglinecounter}
\newcommand{\gk}{\stepcounter{alglinecounter}\arabic{alglinecounter})}

\usepackage{tikz}

\newcommand{\unifdist}{\textsf{Unif}}
\newcommand{\berndist}{\textsf{Bern}}
\newcommand{\prob}{\mathbb{P}}
\newcommand{\mean}{\mathbb{E}}
\newcommand{\ind}{\mathbb{I}}

\begin{document}

\title{Optimal rolling of fair dice using fair coins}

\author{Mark Huber and Danny Vargas}

\maketitle

\begin{abstract}  In 1976, Knuth and Yao presented an algorithm for sampling from a finite distribution using flips of a fair coin that on average used the optimal number of flips.  Here we show how to easily run their algorithm for the special case of rolling a fair die that uses memory linear in the input.  Analysis of this algorithm yields a bound on the average number of coin flips needed that is slightly better than the original Knuth-Yao bound.  This can then be extended to discrete distributions in a near optimal number of flips again using memory linear in the input.
\end{abstract}

\section{Introduction}

A \emph{probabilisitic Turing machine} is a Turing machine that has access to an independent, identically distributed (iid) sequence of flips of a fair coin.  Specifically, say that \( B \) is the flip of a fair coin if \( B \) is equally likely to be 1 or 0.  This means that \( B \) is uniform over \( \{0, 1\} \), written \( B \sim \unifdist(\{0, 1\}) \).  Alternatively, say that \( B \) has the Bernoulli distribution with parameter \( 1 / 2 \), and write \( B \sim \berndist(1 / 2) \).

A natural question to ask is how many flips of a fair coin would it take (on average) to generate a draw from a roll of a fair die?  To be precise, if \( \prob(X = i) = 1 / i \) for \( i \in \{1, \ldots, n \} \) say that \( X \) is the roll of a fair \( n \)-sided die.  Write \( X \sim \unifdist(\{1, \ldots, n\}) \) or \( X \sim \text{d}n \) for short.

In 1976, Knuth and Yao~\cite{knuth1976complexity} developed the optimal algorithm for generating from any discrete probability distribution using a number of fair coin flips. Various authors (such as~\cite{baidya2024efficient, saad2020fast, sinha2013high}) have looked at how to implement Knuth-Yao as efficiently as possible in different use cases.

The purpose of this work is to show how to implement their algorithm for generating from a fair die in a very simple way using a \emph{randomness recycler}~\cite{fill2000randomness, huber2016perfect} approach.

The memory used by the algorithm consists of storing the number \( n \), the original input, storing \( m \) a positive integer which is at most \( 2n - 1 \), storing \( X \) a positive integer at most \( m \), and one bit.  Altogether, this uses \( 3 \lceil \log_2(n) \rceil + 3 \) bits of memory when the input uses \( \lceil \log_2(n) \rceil \) bits of memory.

The algorithm operates as follows.  The idea is to always have a state that is an ordered pair \( (X, m) \) of two random integers.  At any time during the algorithm's run, if the computation conducted by the algorithm used fair coin flips, then \( [X | m] \sim \text{d}m \).  If \( m = n \), then the algorithm is done, since \( X \sim \text{d}n \) as desired.

If \( m < n \), the next step is to flip a fair coin \( B \) and change the state as follows.
\[
(X, m) \mapsto (X + Bm, 2m).
\]
The new state has \( [X + Bm \mid 2m] \sim \text{d}(2m) \).  This allows us to increase the number of sides of the die that \( X \) is a draw from until the number of sides reaches at least \( n \).

At this point, if \( m \geq n \) and \( X \leq n \) then 
\[
(X, m) \mapsto (X, n).
\]
That is to say, accept \( X \) as a draw from \( \text{d}n \).  However, if \( X > n \), do not throw away the state and start over.  Instead, note that \( [X - m \mid m - n] \sim \text{d}(m - n) \).  In other words, what is left over after subtracting the die roll is still a roll of a fair die, just one with a smaller number \( m - n \) of sides.  That is,
\[
(X, m) \mapsto (X - n, m - n).
\]

Return to the doubling procedure and repeat until acceptance occurs.

To illustrate, consider trying to generate a roll from a 5-sided die.  Begin in the state \( (1, 1) \).  Since a 1-sided die only has output 1, \( X = 1 \) is a valid draw from a \( \text{d}1 \).

Now let \( B_1, B_2, B_3, \ldots \) be the iid flips of the fair coin used by the algorithm.  First, the number of sides of the die is doubled to 2, and the current draw from the die is \( 1 + 1 \cdot B_1 \).  So the new state is \( (1 + B_1, 2) \).  Since \( 2 < 5 \), another coin is flipped to get to state \( (1 + B_1 + 2 B_2, 4) \).  Since \( 4 < 5 \), flip another coin to get to state \( (1 + B_1 + 2 B_2 + 4 B_3, 8) \).

Note that \( B_1 + 2B_2 + 4B_3 \) is uniform over \( \{0, 1, 2, \ldots, 7\} \), so \( X = 1 + B_1 + 2 B_2 + 4B_3 \sim \text{d}8 \).  At this point, if \( X \leq 5 \), then this number is uniform over \( \{1, 2, 3, 4, 5\} \) and the algorithm terminates.  If \( X > 5 \), then \( X \) is uniform over \( \{6, 7, 8\} \).  Therefore \( X - 5 \) is uniform over \( \{1, 2, 3\} \), so \( X - 5 \) is a 3-sided die roll.  Therefore the next state is
\[
(X \ind(X \leq 5) + (X - 5) \ind(X > 5), 5 \ind(X \leq 5) + (8 - 5) \ind(X > 5))
\]
where \( \ind(r) \) is the indicator function that evaluates to 1 if expression \( r \) is true, and 0 otherwise.

At this point, if the second component is \( 5 \) the algorithm terminates, otherwise it continues by doubling the number of sides on the die from 3 to 6.  Then it accepts if the state is at most \( 6 \), otherwise returning to state \( (1, 1) \).  The evolution of the state for \( n = 5 \) is given in Figure~\ref{FIG:example}.  At a given state, when there is one child, take that path, when there are two, choose the child uniformly from the two choices.

\begin{figure}[h!]
  \begin{center}
  \begin{tikzpicture}[xscale=0.8]
    \draw (7,0) node (A) {\( (1, 1) \)};
    \draw (3, -1) node (B) {\( (1, 2) \)};
      \draw (1, -2) node (D) {\( (1, 4) \)};
      \draw (5, -2) node (E) {\( (3, 4) \)};
    \draw (11, -1) node (C) {\( (2, 2) \)};
      \draw (9, -2) node (F) {\( (2, 4) \)};
      \draw (13, -2) node (G) {\( (4, 4) \)};

      \draw (0, -3) node (H) {\( (1, 8) \)};
      \draw (2, -3) node (I) {\( (5, 8) \)};

      \draw (4, -3) node (J) {\( (3, 8) \)};
      \draw (6, -3) node (K) {\( (7, 8) \)};

      \draw (8, -3) node (L) {\( (2, 8) \)};
      \draw (10, -3) node (M) {\( (6, 8) \)};

      \draw (12, -3) node (N) {\( (4, 8) \)};
      \draw (14, -3) node (O) {\( (8, 8) \)};

      \draw (0, -4) node (P) {\( (1, 5) \)};
      \draw (2, -4) node (Q) {\( (5, 5) \)};

      \draw (4, -4) node (R) {\( (3, 5) \)};
      \draw (6, -4) node (S) {\( (2, 3) \)};
        \draw(5, -5) node (X) {\( (2, 6) \)};
          \draw(5,-6) node(DD) {\( (2, 5) \)};
        \draw(7, -5) node (Y) {\( (5, 6) \)};
          \draw(7, -6) node (EE) {\( (5, 5) \)};

      \draw (8, -4) node (T) {\( (2, 5) \)};
      \draw (10, -4) node (U) {\( (1, 3) \)};
        \draw(9, -5) node (Z) {\( (1, 6) \)};
        \draw(11, -5) node (AA) {\( (4, 6) \)};
        \draw(9, -6) node (FF) {\( (1, 5) \)};
        \draw(11, -6) node (GG) {\( (4, 5) \)};

      \draw (12, -4) node (V) {\( (4, 5) \)};
      \draw (14, -4) node (W) {\( (3, 3) \)};
        \draw(13, -5) node (BB) {\( (3, 6) \)};
        \draw(15, -5) node (CC) {\( (6, 6) \)};
        \draw(13, -6) node (HH) {\( (3, 5) \)};
        \draw(15, -6) node (II) {\( (1, 1) \)};

    \draw[->] (A) to (B);
    \draw[->] (A) to (C);
    \draw[->] (B) to (D);
    \draw[->] (B) to (E);
    \draw[->] (C) to (F);
    \draw[->] (C) to (G);
    \draw[->] (D) to (H);
    \draw[->] (D) to (I);
    \draw[->] (E) to (J);
    \draw[->] (E) to (K);
    \draw[->] (F) to (L);
    \draw[->] (F) to (M);
    \draw[->] (G) to (N);
    \draw[->] (G) to (O);
    \draw[->] (H) to (P);
    \draw[->] (I) to (Q);
    \draw[->] (J) to (R);
    \draw[->] (K) to (S);
    \draw[->] (L) to (T);
    \draw[->] (M) to (U);
    \draw[->] (N) to (V);
    \draw[->] (O) to (W);
    \draw[->] (S) to (X);
    \draw[->] (S) to (Y);
    \draw[->] (U) to (Z);
    \draw[->] (U) to (AA);
    \draw[->] (W) to (BB);
    \draw[->] (W) to (CC);
    \draw[->] (X) to (DD);
    \draw[->] (Y) to (EE);
    \draw[->] (Z) to (FF);
    \draw[->] (AA) to (GG);
    \draw[->] (BB) to (HH);
    \draw[->] (CC) to (II);
    
  \end{tikzpicture}
  \end{center}
  \caption{From each state, each child is equally likely to be chosen.  When a state of the form \( (i, 5) \) is reached, output \( i \).  When state \( (1, 1) \) is reached, begin again at the root.}
  \label{FIG:example}
\end{figure}
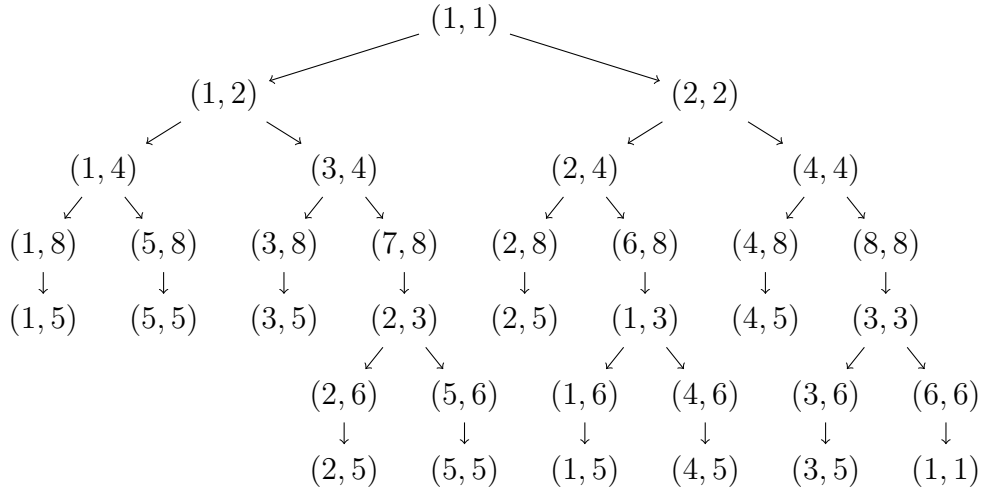
  
In general, the pseudocode for this algorithm is as follows.

\vspace*{12pt}
\begin{tabular}{rl}
    \toprule
    \multicolumn{2}{l}{\textbf{Algorithm} \textsf{optimal\_uniform}(n) } \\
    \midrule
  \gk & \hspace{0em} \( m \leftarrow 1,\ X \leftarrow 1 \) \\
  \gk & \hspace*{0em} while \( m < n \) \\
  \gk & \hspace*{1em} draw \( B \) uniform over \( \{0, 1\} \) \\
  \gk & \hspace*{1em} \( X \leftarrow X + Bm \), \( m \leftarrow 2m \) \\
  \gk & \hspace*{1em} if \( m \geq n \) and \( X \leq n \) then \( m \leftarrow n \) \\
  \gk & \hspace*{1em} if \( m \geq n \) and \( X > n \) then \( X \leftarrow X - n \), \( m \leftarrow m - n \) \\ 
  \gk & \hspace*{0em} return \( X \) \\  
  \bottomrule
\end{tabular}
\vspace*{12pt}

This type of algorithm, where either acceptance occurs, or rejection occurs and as much of the randomness leftover is used in the next state is called the \emph{randomness recycler} protocol~\cite{fill2000randomness, huber2016perfect}.

This method works best with a fair die.  The \emph{Alias method}~\cite{walker1977efficient} gives a close to efficient method of generating from a loaded die, that is, an arbitrary distribution over a finite number of states.  The new method presented here can be extended to handle this situation as well, but it is less simple to implement.

The rest of the work is organized as follows.  The next section contains a proof that the algorithm is correct and that it uses the optimal number of coin flips.  Section~\ref{sec:three} shows an upper bound on the expected number of flips needed that is slightly better than that of~\cite{knuth1976complexity}.  
The last section then shows how this technique can be extended to loaded dice.

\section{Proof of correctness}

\begin{theorem}
\label{thm:one}
    The output of \( \textsf{optimal\_uniform}(n) \) is uniform over \( \{1, \ldots, n \} \).
\end{theorem}

It helps to have the following fact in place.

\begin{lemma}
    At the end of each line of the algorithm, for state \( (X, m) \) it holds that \( [X \mid m] \sim \text{d}m \).
\end{lemma}

\begin{proof}
    The state \( (X, m) \) is changed by lines 1, 4, 5, and 6.  When \( (X, m) = (1, 1) \), it is trivially true that \( X \sim \text{d}1 \), so line 1 maintains this invariant.

    Suppose \( (X, m) \) has \( X \sim \text{d}m \) before running line 4, 5, or 6.  For line 4, if \( X \sim \unifdist(\{1, \ldots, m\} \) then \( X + m \sim \unifdist(\{m + 1, \ldots, 2m\} \).  Since \( B \) is equally likely to be \( 0 \) or \( 1 \), \( X + Bm \) is equally likely to be any element of \( \{1, \ldots, 2m \} \) and the invariant is maintained.

    For line 5, it is a well known fact that for \( X \sim \text{d}m \) where \( m \geq n \), \( [X | X \leq n] \sim \text{d}n \).  Since \( X \leq n \) leads to \( m = n \), \( [X | m = n] \sim \text{d}n \), Similarly for line 6, if \( X \sim \text{d}m \) where \( m > n \), \( [X \mid X > n] \sim \unifdist(\{n + 1, \ldots, m \} \).  So \( [X - n \mid X > n] \sim \unifdist(\{1, \ldots, m - n\} \).
    
    An induction then completes the proof.
\end{proof}

\begin{proof}[Proof of Theorem~\ref{thm:one}]
The theory of the \emph{randomness recycler} (see for instance~\cite{fill2000randomness,huber2016perfect}) says that two conditions are necessary and sufficient to establish correctness.  The first is that \( [X | m] \sim \text{d}m \) throughout the running of the algorithm, which is true from the previous lemma.  The second is that the algorithm terminates with probability 1.  After at most \( \lceil \log_2(n) \rceil \) steps there will have been at least one opportunity for \( m \) to grow to be at least \( n \).  When \( m \in \{n, n + 1, \ldots, 2n - 1\} \), there is at least a \( 1 / m \geq 1 / (2n) \) chance of termination.  Hence the probability that the algorithm terminates approaches 1 as the number of steps goes to infinity.
\end{proof}

\section{The optimal DDG}

In fact, this simple algorithm is not only correct, it is optimal in the following sense.  

Consider any algorithm that utilizes random bits to reach an outcome.  This can be represented using a \emph{discrete distribution generating} tree, called a DDG for short.  This is a rooted binary tree where each leaf of the tree is associated with an outcome.  The algorithm works by starting with the current node equal to the root of the tree.  A coin is flipped, if the result is 0 move to the left child of the current node, otherwise move to the right.  If the resulting node is a leaf, output the outcome associated with that leaf.  Otherwise make this the new current node, and repeat.

Figure~\ref{fig:DDG} shows two DDGs that output a random variable with probabilities \( (3 / 8, 1 / 2, 1 / 8) \) over outcomes \( (1, 2, 3) \).  The tree on the right will be preferred in practice because the average depth (and hence average number of coin flips used by the algorithm) is smaller.

\begin{figure}[ht!]
\centering
    \begin{tikzpicture}
      \draw (0, 0) to (-1, -1);
      \draw (0, 0) to (1, -1);
      \draw (-1, -1) to (-1.5, -2);
      \draw (-1, -1) to (-0.5, -2);
      \draw (1, -1) to (1.5, -2);
      \draw (1, -1) to (0.5, -2);
      \draw (-1.5, -2) to (-1.75, -3) node[below] {1};
      \draw (-1.5, -2) to (-1.25, -3) node[below] {1};
      \draw (-0.5, -2) to (-0.25, -3) node[below] {1};
      \draw (-0.5, -2) to (-0.75, -3) node[below] {2};
      \draw (1.5, -2) to (1.75, -3) node[below] {2};
      \draw (1.5, -2) to (1.25, -3) node[below] {2};
      \draw (0.5, -2) to (0.25, -3) node[below] {2};
      \draw (0.5, -2) to (0.75, -3) node[below] {3};
    \end{tikzpicture} 
    \hspace*{12pt}
    \begin{tikzpicture}
      \draw (0, 0) to (-1, -1) node[below] {2};
      \draw (0, 0) to (1, -1);
      \draw (1, -1) to (1.5, -2);
      \draw (1, -1) to (0.5, -2) node[below] {1};
      \draw (1.5, -2) to (1.75, -3) node[below] {1};
      \draw (1.5, -2) to (1.25, -3) node[below] {3};
    \end{tikzpicture} 
    \caption{Two DDGs for \( (3 / 8, 1 / 2, 1 / 8) \).  The DDG on the left will always use 3 flips to determine the outcome, while the one on the right uses 1 flip with probability \( 1 / 2 \), 2 flips with probability \( 1 / 4 \), or 3 flips with probability \( 1 / 4 \).}
    \label{fig:DDG}
\end{figure}
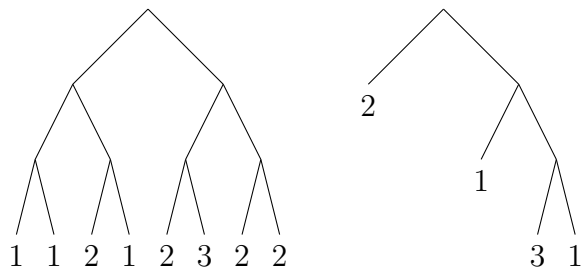

Knuth and Yao defined an optimal DDG for a probability distribution \( p = (p_1, \ldots, p_n) \) as follows.  Let \( N(T) \) denote the random number of bits used by DDG \( T \).  For \( p \), say that \( T \) is \emph{optimal} if for any other DDG \( T' \) with output \( p \), for all \( i \), \( \prob(N(T) > i) \leq \prob(N(T') > i) \). 

Another way to say this is that the number of flips used by the optimal tree is stochastically dominated by the number of flips used by any other tree.

As usual, define the \emph{level} of a node to be the number of edges in the path from the root to the node.  Knuth and Yao showed the following result~\cite{knuth1976complexity}.
\begin{theorem}
    For probability vector \( (p_1, \ldots, p_n) \), a DDG is optimal if and only if outcome \( i \) appears on level \( j \) 
    a number of times equal to 
    \[
    \lfloor 2^j p_i \rfloor - 2 \lfloor 2^{j - 1} p_i \rfloor
    \] (this is the \( j \)th bit in the binary expansion of \( p_i \), using the finite expansion when possible.)
\end{theorem}

This is equivalent to the following.
\begin{lemma}
    For probability vector \( (p_1, \ldots, p_n) \), a DDG is optimal if and only if outcome \( i \) appears on any given level either 0 or 1 times, and there is no \( k \) such that outcome \( i \) appears at level \( k, k + 1, k + 2, \ldots \).  
\end{lemma}

\begin{proof}
    If outcome \( i \) appears once at level \( j \), this leaf contributes \( 1 / 2^j \) to the probability of outcome \( i \).  Since \( p_i \) has a unique binary expansion where there are not an infinite number of trailing 1's, this DDG must be equivalent to \( i \) appearing on level \( i \) if and only if the (finite if possible) binary expansion of \( p_i \) has a 1 at level \( i \).  
\end{proof}

This characterization can show that \( \textsf{optimal\_uniform}(n) \) is optimal.

\begin{theorem}
    The algorithm \( \textsf{optimal\_uniform}(n) \) yields an optimal DDG.
\end{theorem}

\begin{proof}

    Consider a state \( (X, m) \) at a level greater than 1 in the tree.  Then there is exactly one state \( (X', m') \) and coin flip \( B' \) that moves to this state.

    If the previous level did not have any leaves, then \( m' = m / 2 \), \( B' = \ind(X > m') \), and 
   \[
    X' = X - m' \ind(X > m')
    \]
    is the only state and flip that leads to \( (X, m) \).

    If the previous level did have leaves, then the state that moves at line 6 to \( (X, m) \) is \( (X + n, m + n)  \).  Therefore, the state at the previous level that moves to this state is \( m' = (m + n) / 2 \), \( B' = \ind(X > m') \), and 
    \[
    X' = X + n - B m'
    \]
    is the only option for the previous state and flip that leads to the current state.

    But since the initial state is fixed at \( (1, 1) \), this forms the basis of an induction proof that any state \( (X, m) \) at a later level is determined by a unique set of flips.  In particular, if \( (X, m, B) \) moves to a leaf at a particular level, then there is a unique set of flips that led to that state.  That means that for any outcome at a particular level, there is exactly one leaf at that level that results in that state.  Therefore each outcome appears on at most one leaf at each level.

    Now consider how many levels in a row the DDG can have a particular outcome as a leaf.  A leaf only occurs when \( m \geq n \) and \( X \leq n \) at line 6.  If \( X > n \), then at the next state \( m \leftarrow m - n \).  

    Therefore, if there are \( r \) values for \( X \) that reject at that level, and at the next level there are at most \( 2r - n \) states of \( X \) that can reject at line 6.

    At the next level there will be at most \( 4r - 2n - n = 4r - 3n \) values of \( X \) that reject, then \( 8r - 7n \), then \( 16r - 15n \) and so on.  After \( k \) levels there will be \( 2^k r - (2^k - 1)n = 2^k(r - (1 - (1 / 2^k)n \) states of \( X \) that lead to rejection.  Because \( r < n \), this either reaches 0 or goes negative, indicating that rejection does not occur.  If rejection does not occur at a particular level then there cannot be a leaf at that same level.

    Therefore, it cannot be the case that the DDG for this algorithm has an infinite sequence of consecutive levels with leaves.  Therefore, by the previous lemma this DDG must be optimal.
    
\end{proof}

\section{Upper bounding the expected number of flips}
\label{sec:three}

It is straightforward to calculate exactly the average number of flips used by \( \textsf{optimal\_sampling}(n) \) for any particular \( n \).

For instance, if \( n = 5 \) it takes three flips to make \( m = 8 \), at which point there is a \( 5 / 8 \) chance of accepting the result and terminating the algorithm.  If it does not accept then \( m = 8 - 5 = 3 \), and one more flip makes \( m = 6 \).  At this point either the algorithm accepts with probability \( 5 / 6 \) or rejects and sets \( m = 6 - 5 = 1 \), so the algorithm starts from the beginning.  Hence
\[
\mean[N_{n = 5}] = 3 + \frac{3}{8}\left[1 + \frac{1}{6} \mean[N_{n = 5}] \right],
\]
so \( \mean[N_{n = 5}] = 3.6 \).

Note that any DDG for a fair die will need at least \( \lceil \log_2(n) \rceil \) bits on average.  Results in~\cite{knuth1976complexity} give that \( \log_2(n) + 2 \) bits suffice on average.  The following result strengthens this slightly.

\begin{theorem}
    The expected number of flips used by \( \textsf{optimal\_sampling}(n) \) is at most \( \lceil \log_2(n) \rceil + 1 \).
\end{theorem}

\begin{proof}
    Start the algorithm by flipping coins until the initial state \( m_1 \) has doubled to be at least \( n \).  This will require \( f_1 = \lceil \log_2(n / m) \rceil \) flips.  After these flips, say that the state is \( (X_2, m_2) \) where \( m_2 \in \{n, \ldots, 2n - 1\} \).  To be specific,
    \[
    m_2 = 2^{\lceil \log_2(n / m) \rceil}.
    \]

    The chance of accepting is \( n / m_2 \), so the chance that further work is needed is \( 1 - n / m_2 \).  If further work is needed, then \( f_2 \) more flips are needed to double \( m_2 \) back up to be at least \( n \), and so on.  Altogether the expected number of flips \( t \) is 
    \[
    t = f_1 + \left(1 - \frac{n}{m_2} \right) \left[ f_2 + \left(1 - \frac{n}{m_3}\right)\left[f_3 + \left( 1 - \frac{n}{m_4} \right) \cdots \right.\right.
    \]
    where \( f_k = \lceil \log_2(n / (n - m_{k})) \rceil \) for \( k \geq  2 \).
    
    Suppose this is thought of as a function of \( m_2, m_3, \ldots \) chosen independently of each other in \( \{n, \ldots, 2n - 1\} \).  Then once \( m_2 \) is chosen to maximize the result, \( m_3 \) should receive the same value because of the self-similarity of the expression. 
 Hence there is \( m^* \in \{n, \ldots, 2n - 1\} \) such that 
    \[
    t \leq f_1 + \left(1 - \frac{n}{m^*} \right) \left[ f_{m^*} + \left(1 - \frac{n}{m^*}\right)\left[f_{m^*} + \left( 1 - \frac{n}{m^*} \right) \cdots \right.\right.
    \]

    Setting 
    \[
    \alpha = \left(1 - \frac{n}{m^*} \right) \left[ f_{m^*} + \left(1 - \frac{n}{m^*}\right)\left[f_{m^*} + \left( 1 - \frac{n}{m^*} \right) \cdots \right.\right.
    \]
    gives
    \[
    \alpha = \left(1 - \frac{n}{m^*} \right) \left[ f_{m^*} + \alpha \right],
    \]
    so
    \[
    \alpha = \left( \frac{m^*}{n} - 1 \right) \left(\lceil \log_2\left(\frac{n}{m^* - n} \right) \right\rceil.
    \]

    Let \( \beta = (m^* - n) / n \).  If \( \beta \in (1 / 2, 1] \), then \( \alpha = 1 \cdot 1 = 1 \).  If \( \beta \in (1 / 4, 1/ 2] \), then \( \alpha = (1 / 2) \cdot 2 = 1 \).  If \( \beta \in (1 / 8, 1 / 4] \), \( \alpha = (1 / 4) \cdot 3 = 3 / 4 \), and for smaller \( \beta \) the value of \( \alpha \) only shrinks further.

    Hence \( t \leq \lceil \log_2(n) \rceil + 1 \).
    
\end{proof}

\section{Extending to discrete distributions}
\label{sec:four}

For the general algorithm, the idea is to also keep state \( (X, m) \), but use \( X \) to draw uniformly from the set of outcomes that have a \( i \) in their bit description at the bit equal to the number of flips of the coin taken.  

\vspace*{12pt}
\begin{tabular}{rl}
    \toprule
    \multicolumn{2}{l}{\textbf{Algorithm} \( \textsf{optimal\_discrete}(p = (p_1, p_2, \ldots)) \) } \\
    \midrule
  \gk & \hspace{0em} \( m \leftarrow 1,\ X \leftarrow 1, Y \leftarrow 0 \) \\
  \gk & \hspace*{0em} while \( Y = 0 \) \\
  \gk & \hspace*{1em} let \( A \leftarrow\{i:2p_i \geq 1\} \),
  \( n \leftarrow \#(A) \) \\
  \gk & \hspace*{1em} for all \( i \), let \( p_i \leftarrow 2p_i - \ind(2p_i \geq 1)\) \\
  \gk & \hspace*{1em} draw \( B \) uniform over \( \{0, 1\} \) \\
  \gk & \hspace*{1em} \( X \leftarrow X + Bm \), \( m \leftarrow 2m \) \\
  \gk & \hspace*{1em} if \( m \geq n \) and \( X \leq n \) then \( m \leftarrow n \), let \( Y \) be the \( X \)th element of \( A \).  \\
  \gk & \hspace*{1em} if \( m \geq n \) and \( X > n \) then \( X \leftarrow X - n \), \( m \leftarrow m - n \) \\ 
  \gk & \hspace*{0em} return \( Y \) \\  
  \bottomrule
\end{tabular}
\vspace*{12pt}

\begin{theorem}
    Algorithm \( \textsf{optimal\_discrete}(p) \) is optimal.
\end{theorem}

\begin{proof}
    Although the value of \( n \) changes from flip to flip, the arguments used in the uniform case remain the same.  Here there is a unique set of flips that results in a particular state  \( (X, m) \) and hence any particular leaf outcome appears only once at a level.  

    Here \( A \) is chosen after \( d \) flips so that \( i \in A \) if and only if \( \lceil 2^d A_i \rceil = 1 \), so the DDG must be optimal.
\end{proof}

Of course this algorithm is a bit more abstract than the first, given that to actually implement line 3 requires an analysis of the probability distribution being used to determine which of the outcomes have a 1 bit in the position equal to the number of flips.

\section{Conclusion}

The Knuth-Yao optimal sampler can be implemented using a randomness recycler type algorithm.  This allows for a very simple implementation of the algorithm for generating the roll of a fair die and can be extended to handle general discrete distributions given the ability to draw out the \( i \)th bit in the binary expansion of the probabilities.

\printbibliography

\end{document}